\documentclass[11pt]{article}

\usepackage{amsmath,rotating,subfigure}

\usepackage{times}
\usepackage{bm}
\usepackage{natbib}
\usepackage{algorithm}
\usepackage{xargs}[2008/03/08]
\usepackage{amssymb}
\usepackage{mathrsfs}
\usepackage{graphicx}
\usepackage{tikz}
\usepackage{soul}
\usepackage{smile}
\usepackage[flushleft]{threeparttable}

\usepackage{multirow}
\usepackage[colorlinks,
linkcolor=red,
anchorcolor=blue,
citecolor=blue
]{hyperref}

\def\T{{ \mathrm{\scriptscriptstyle T} }}

  \newcommand{\bea}{\begin{eqnarray*}}
  	\newcommand{\eea}{\end{eqnarray*}}
  \newcommand{\ee}{\end{eqnarray}}
  \newcommand{\bay}{\begin{array}}
  	\newcommand{\eay}{\end{array}}
  \newcommand{\ben}{\begin{enumerate}}
  	\newcommand{\een}{\end{enumerate}}
  \newcommand{\bcen}{\begin{center}}
  	\newcommand{\ecen}{\end{center}}
  
\def\##1\#{\begin{align}#1\end{align}}
\def\$#1\${\begin{align*}#1\end{align*}}

\newcommand{\rF}{\textnormal{F}}
\renewcommand{\tr}{\textrm{trace}}

\global\long\def\expect{E}
\global\long\def\prob{\mathrm{Pr}}

\global\long\def\real{\mathbb{R}}

\global\long\def\manifold{\mathcal{M}}

\global\long\def\vec#1{\overrightarrow{#1}}
\newcommandx\pf[2][addprefix=\global]{{#1}^\prime_{#2}}

\newcommandx\estneighbor[2][usedefault, addprefix=\global, 1=\delta, 2=x]{\hat{\mathcal{N}}_{\mathcal{L}^{2}}(#1,#2)}
\newcommandx\estneighborM[2][usedefault, addprefix=\global, 1=\delta, 2=x]{\mathcal{N}_{\manifold}(#1,#2)}
\newcommandx\estneighborindex[2][usedefault, addprefix=\global, 1=\delta, 2=x]{\hat{\mathcal{I}}(#1,#2)}
\newcommandx\ambientball[1][usedefault, addprefix=\global, 1=\delta]{\mathbb{B}_{#1}^{\mathcal{L}^{2}}(x)}

\newcommandx\manball[1][usedefault, addprefix=\global, 1=\delta]{\mathbb{B}_{#1}^{\manifold}(x)}

\global\long\def\Exp{\mathrm{Exp}}
\global\long\def\spd{\mathrm{Sym}_{\star}^{+}(m)}
\newcommandx\tangentspace[2][usedefault, addprefix=\global, 1=\manifold]{T_{#2}#1}
\global\long\def\sym{\mathrm{Sym}(m)}
\newcommandx\fronorm[2][usedefault, addprefix=\global, 1=]{\|#2\|_{\rF}^{#1}}
\newcommandx\lnorm[2][usedefault, addprefix=\global, 1=]{\|#2\|_{2}^{#1}}
\global\long\def\innerprod#1#2{\langle#1,#2\rangle}

\newcommandx\froinnerprod[3][usedefault, addprefix=\global, 1=]{\innerprod{#2}{#3}_{\rF}^{#1}}

\def\T{{ \mathrm{\scriptscriptstyle T} }} 

\newcommand{\Rom}[1]{\text{\uppercase\expandafter{\romannumeral #1\relax}}}

\usepackage{geometry}
\usepackage{enumitem}

\begin{document}

\title{ \LARGE Modeling Symmetric Positive Definite Matrices with An Application to Functional Brain  Connectivity}     


\author{
Zhenhua Lin\thanks{Department of Statistics, University of California, Davis.},~~Dehan Kong\thanks{Department of Statistical Sciences, University of Toronto.} ~and ~Qiang Sun\thanks{Address for correspondence: Department of Statistical Sciences, University of Toronto,
Toronto, Ontario M5S 3G3, Canada; {\textsc E-mail:} \texttt{qsun@utstat.toronto.edu}.}
}


\date{ }

\maketitle

\vspace{-0.25in}

\begin{abstract}
In neuroscience, functional brain connectivity describes the connectivity between brain regions that share functional properties. Neuroscientists often characterize it by a time series of covariance matrices between functional measurements of distributed neuron areas. An effective statistical model for functional connectivity and its changes over time is critical for better understanding the mechanisms of brain and various neurological diseases. To this end, we propose a  matrix-log mean model with an additive heterogeneous noise  for modeling random symmetric positive definite matrices that lie in a Riemannian manifold. The heterogeneity  of error terms is introduced specifically  to  capture the curved nature of the manifold. We then propose to use  the local scan  statistics to detect change patterns in the  functional connectivity.  Theoretically, we show that our procedure can recover all change points consistently. Simulation studies and an application to the Human Connectome Project lend further support to the proposed methodology. 

\end{abstract}
\noindent
{\bf Keywords}: change point, functional connectivity,  Riemannian manifold, sure coverage property.

\section{Introduction}\label{sec:1}

Understanding the functional brain connectivity  is critical for understanding the fundamental mechanisms of brain, how it works, and various neurological diseases. It  has attracted great interest recently.  For instance, the Human Connectome Project  investigates the structural and  functional connectivity in order to diagnose cognitive abilities of individual subjects. Functional connectivity can be defined as 	temporary statistical dependence between spatially remote neurophysiological events \citep{friston2011functional}, 
and has been observed to be dynamic in nature, 
even in the resting state \citep{hutchison2013dynamic}. In practice, neuroscientist often characterize the dynamic functional connectivity by a series of symmetric positive definite (SPD) covariance matrices between functional measurements of neuronal activities across different regions in human brain. Establishing appropriate dynamic  models is critical for understanding fundamental mechanisms of brain networks  and has attracted much  attention  in  neuroscience recently \citep{xu2015dynamic,hutchison2013dynamic}. 


However, little has been done in the statistics community for investigating dynamic changes of functional connectivity over time. 
The non-Euclidean structure of covariances has  introduced significant challenges to the development of proper statistical models and their analysis. Indeed, all SPD matrices form a nonlinear Riemannian manifold, which is referred to as the SPD manifold.  Motivated by the SPD manifold structure under the Log-Euclidean metric  \citep{Arsigny2007}, we use the matrix logarithm to embed the SPD matrices into a Hilbert space -- an Euclidean space up to a symmetric structure, to be concrete. We then model the transformed random SPD matrix using a mean model with an additive heterogeneous noise.  The heterogeneous error depends on the tangent space of the mean SPD matrix and thus takes  the curved structure of the SPD manifold into account.  Our work refines the previous work by   \citet{chiu1996matrix}, whose model does not respect the original manifold structure.   

Built on this statistical model,  we then propose to use a form of local  scan statistics to detect multiple change patterns that are present in the functional brain connectivity over time.  To the best of our knowledge, ours is the first work  on the study of change point detection for SPD manifold-valued data. 
Although the proposed method is primarily motivated by discovering change patterns in fMRI, 
it has the potential to be applied to many other applications, such as diffusion tensor imaging \citep{dryden2009non} and longitudinal data analysis \citep{daniels2002bayesian}.




\subsection{Related Literature}
Change point detection  with at most a single change point has been widely studied in the literature. When the distributions of the data are assumed to be known, score- or likelihood-based procedures can be applied \citep{james1987tests}. Bayesian and nonparametric approaches have also been proposed, see \cite{carlstein1994change} for a review. More recently,  \cite{chen2015graph} proposed a graph based approach for  nonparametric change point detection.    When there are multiple change points, the problem  becomes much more complicated. Some popular approaches include the exhaustive search with Schwarz criterion \citep{yao1988estimating}, the circularly binary segmentation \citep{olshen2004circular} and the fused lasso method \citep{tibshirani2007spatial}. 
In genomics, these techniques have been exploited to study DNA copy number variations, 
see \cite{olshen2004circular, zhang2007modified, tibshirani2007spatial} among others. 
However, 
none of the above methods  deals with Riemannian data. 

 There have  been a few works on Riemannian data analysis in the statistics literature. 
For example, \citet*{Schwartzman2006} proposed 
several test statistics for comparing the means of two populations of symmetric positive definite matrices.  \citet{Zhu2009}  developed a semiparametric regression model for symmetric definite positive matrices with Euclidean covariates. Later,   \citet{YuanYing2012} studied the local polynomial regression in the same setting. 
\citet{steinke2009non} consider nonparametric regression between general Riemannian manifolds. 
\citet*{Petersen2016} developed a novel Fr\'echet regression approach for complex random objects with euclidean covariates. We believe our work will be a valuable addition to the literature.  

\section{Geometric Interpretation}\label{sec:2}


We briefly introduce $\spd$, the Riemannian manifold consisting of all $m\times m$ symmetric positive definite matrices, while we refer readers to the appendix for more details. A Riemannian manifold is a smooth manifold endowed with an inner product
$\innerprod{\cdot}{\cdot}_{x}$ on the
tangent space at each point $x$, such that $\innerprod{\cdot}{\cdot}_{x}$
varies with $x$ smoothly.   
We consider the Log-Euclidean metric  for the symmetric positive definite matrix manifold due to  its computational tractability \citep{Arsigny2007}. Other metrics include the naive Frobenius metric which does not account for the curved nature of symmetric positive definite matrices and the affine invariant metric  which is more difficult to compute \citep{Terras2012}. 


For a manifold $\manifold$, we use  $\tangentspace[\mathcal{M}]x$ to denote the tangent space at the base point $x$. 
{It can be shown that the tangent space $\tangentspace[\spd]I$ at the identify matrix $I$ is the space of $m\times m$
symmetric matrices, denoted by $\sym$.}   
For the Log-Euclidean metric, the inner product between $U,V\in\sym$ on $\tangentspace[\spd]I$ at the identity matrix $I$ is defined as  $\innerprod UV_I=\tr(U V)$. To define the inner product at a general point, we utilize the concept of differential maps.  For a smooth transformation $\varphi:\mathcal{N}\rightarrow\manifold$ between two manifolds, 
its {differential}
at $x$, denoted by $\pf{\varphi}{x}$, is a linear map sending a tangent vector  $v\in\tangentspace[\mathcal{N}]x$
to a tangent vector $\pf{\varphi}{x}(v)\in\tangentspace[\manifold]{\varphi(x)}$.
{See Figure \ref{fig:differential} for a graphical illustration.}
 When both $\mathcal{N}$ and $\manifold$ are Euclidean submanifolds, the differential $\pf{\varphi}{x}$ is the usual notion of differential of the function $\varphi$ at $x$, given by a Jacobian matrix. With this formalism, we consider the smooth map $\log: \spd\rightarrow \sym$, where $\log$ is the matrix logarithm, the inverse map of  the matrix exponential. The matrix exponential of a matrix $U\in \sym$ is defined as $\exp(U)=\sum_{k=0}^\infty U^k/k!$.  
 The Riemannian metric at a general point $S$ is then defined as  $\innerprod UV_S=\langle \log'_S U, \log'_S V\rangle_I$, where $\log'_S: \tangentspace[\spd]S\rightarrow \tangentspace[\sym]S$ is a linear  operator \citep{Arsigny2007}. The Riemannian exponential map under this metric
is given by $\Exp_{S}U=\exp(\log S+\pf{\log}{S}U)$, where $\exp$ is the matrix exponential.  Riemannian exponential maps are closely related to the intrinsic properties of a manifold, such as the geodesics and the Gauss curvature \citep{Lee1997}. 

\begin{figure}[t]
\center{
\includegraphics[scale=.78]{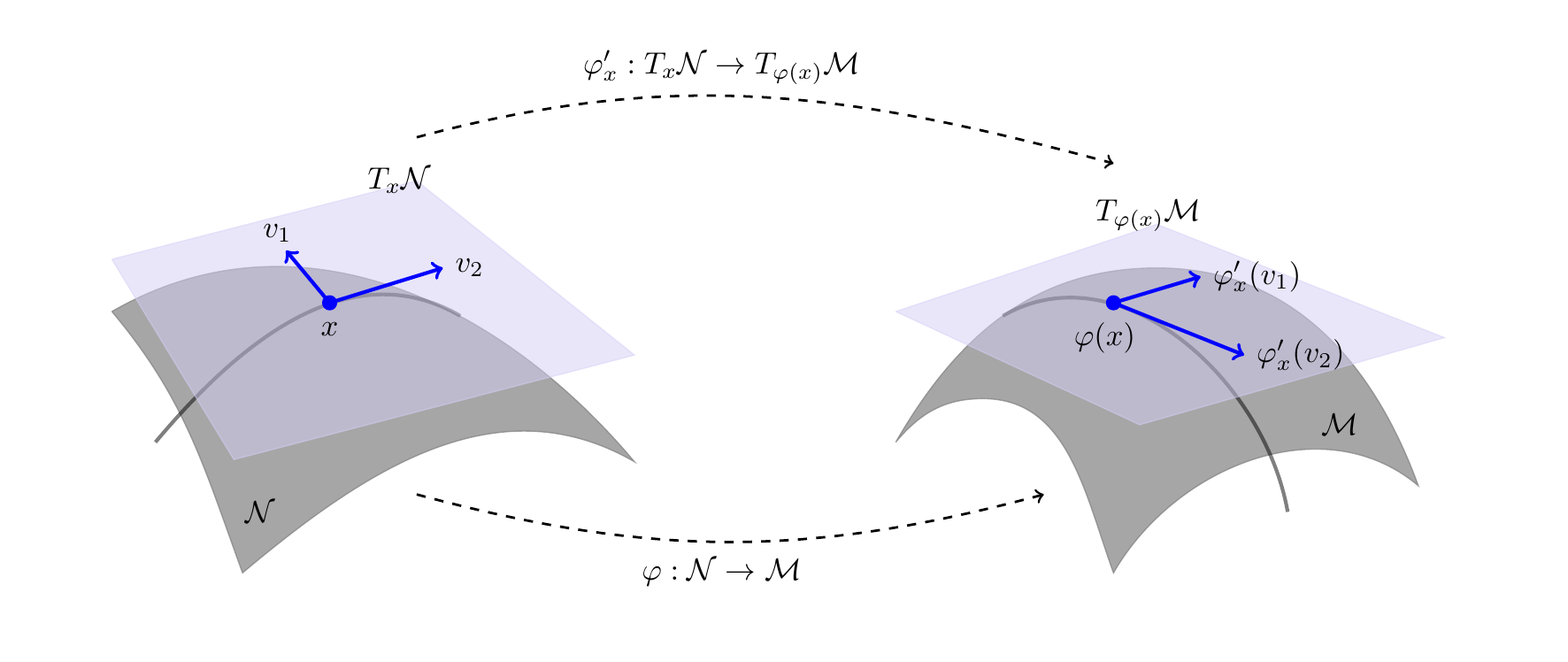}
}
\caption{Illustration of smooth map between manifolds and its differential.}\label{fig:differential}
\end{figure}

\section{Methodology}\label{sec:3}

\subsection{A Heterogeneous Matrix-log Mean Model  \label{subsec:Model}}

Suppose   
we have collected a sequence of  matrix-valued observations $Y_{1},\ldots,Y_{n}\in \spd$. 
 We propose the following matrix-log mean model for investigating the mean changes of the data sequence
\#
\log Y_{i}=\log\mu_{i}+\pf{\log}{\mu_i}\varepsilon_{i}, \label{eq:model-sym}
\#
where  $\mu_{i}\in\spd$ is the mean matrix, and $\varepsilon_{i}\in\tangentspace[\spd]{\mu_{i}}$
is a {mean-zero  error term in $\tangentspace[\spd]{\mu_{i}}$}. 
Here  $\pf{\log}{\mu_i}: \tangentspace[\spd]{\mu_{i}} \rightarrow \tangentspace[\sym]{\log \mu_i}$ is a linear operator
acting on $\varepsilon_{i}$. 
The noise term $\pf{\log}{\mu_i}\varepsilon_{i}$ has mean zero, but the corresponding covariance   
depends on $\mu_{i}$.  Hence model \eqref{eq:model-sym} has a heterogeneous noise
component. 

Interestingly, the heterogeneity of the noise terms makes use of the Riemannian manifold structure introduced in Section \ref{sec:2}. Without using the geometric structure, one could simply apply the matrix logarithm first and then model the random SPD matrices $Y_i$'s as 
\#
\log Y_{i}=\log\mu_{i}+\xi_i, \notag 
\#
where $\xi_i$'s are identically distributed random elements. This naive model, first introduced by \cite{chiu1996matrix}, misses the curved structure in the SPD manifold, and thus is less efficient for estimation and inference. 
{Different from theirs, we introduce  the {location-dependent transformations $\log_{\mu_i}^\prime$}  in model \eqref{eq:model-sym}   
to  respect the original manifold structure, because it turns this model  into a geodesic/intrinsic mean model}. To appreciate this, we take matrix exponential on both sides and find that 
\#
Y_i=\exp(\log\mu_i+\log'_{\mu_i}\varepsilon_i)=\Exp_{\mu_i}\varepsilon_i.\notag
\#
It can shown that $\mu_i$ is the minimizer to the following optimization program
\$
\mu_i=\argmin_{S\in\spd} \EE g^2(S, Y_i),
\$
where $g(S, Y_i)$ is the geodesic distance between $A$ and $Y_i$ in $\spd.$ Therefore, model \eqref{eq:model-sym} serves as an exact counterpart of the Euclidean  mean model $Y=\mu+\varepsilon\in \RR^d$, where $\mu$ minimizes $\EE \|Y-a\|_2^2$ over $a\in \RR^d$.

\begin{remark}
We emphasize  that the idea of using a matrix logarithm  to  model SPD matrices  was first explored by \citet{leonard1992bayesian} and \citet{chiu1996matrix}. 
However, their approach does not take  the manifold structure into account. From the modeling perspective, our key contribution is that we establish a parametric model for SPD matrices that respects the original manifold structure. 
\end{remark}

Model \eqref{eq:model-sym} provides a natural way to investigate change-point detection problems for SPD manifold-valued data. For this purpose, we further 
 assume that there exist $\mathcal{J}=\{\tau_{j}:j=1,\ldots,J\}$ and $1\leq\tau_{1}<\cdots<\tau_{J}\leq n-1$
such that $\mu_{\tau}\neq\mu_{\tau+1}$ if $\tau\in\mathcal{J}$ and
$\mu_{\tau}=\mu_{\tau+1}$ otherwise. Elements in $\mathcal{J}$ are
called change points. Our goal  is to detect $\mathcal{J}$ based on the data sequence $Y_{1},\ldots,Y_{n}$.


\subsection{Computational Details}
Computationally, it is more convenient to work with a basis of the space $\sym$ which is a $d=m(m+1)/2$ dimensional Hilbert space under
the Frobenius inner product. The Frobenius inner product between $A=(a_{ij})$ and $B=(b_{ij})$  is defined as $\froinnerprod AB=\sum_{i,j=1}^{m}a_{ij}b_{ij}$.  Let $\phi=\{\phi_k:\,1\leq k\leq d \}$ be an orthonormal basis of $\sym$ under this inner product. 
Then, for any $A\in\sym$,
we can write $A=\sum_{k=1}^{d}c_{k}\phi_{k}$ with $c_k=\froinnerprod{A}{\phi_k}$, 
and identify it with its coefficient vector $(c_1,\ldots,c_d)^\T$, denoted by $\vec{A}$. 

In this paper, we adopt the basis constructed in the following.  
Let $B_{ij}$ be the matrix
of zeros except the $(i,j)$ and $(j,i)$ entries, which are set to $1$ if $i=j$, and   $1/\sqrt{2}$ otherwise. Since $B_{ij}=B_{ji}$,
we consider  basis matrices $B_{ij}$'s with $i\geq j$. It can be
checked that $\fronorm{B_{ij}}=1$ and $\froinnerprod{B_{ij}}{B_{k\ell}}=0$
if $i\neq k$ or $j\neq\ell$, where $\fronorm{\cdot}$ denotes the
Frobenius norm. Let  $\phi_{i(i-1)/2+j}=B_{ij}$, and then $\{B_{ij}:1\leq j\leq i\leq m\}$
form an orthonormal basis for $\sym$. We use this basis in our computation. Note that the results presented in the paper are identical for all bases. 

To compute the matrix logarithm $\log Y$ for $Y\in\spd$, we first
find a unitary matrix $P$ such that $Y=P\Lambda P^{-1}$ for a diagonal
matrix $\Lambda=\mathrm{diag}(\lambda_{1},\ldots,\lambda_{m})$. The
matrices $P$ and $\Lambda$ can be computed by eigendecomposition
or singular value decomposition (SVD). Then $\log Y=P\log(\Lambda)P^{-1}$
with $\log\Lambda=\mathrm{diag}(\log\lambda_{1},\ldots,\log\lambda_{m})$.

To compute the matrix representation of the linear differential operator $\pf{\log}{\mu}$
for a given symmetric positive definite matrix $\mu$ with respect to the basis $\boldsymbol{\phi}$,
we first note that $\pf{\log}{\mu}=(\pf{\exp}{\log\mu})^{-1}$ \citep{Arsigny2007}. Therefore,
once we have the matrix representation $Q$ of the linear operator
$\pf{\exp}{\log\mu}$ with respect to the basis $\phi$, then $Q^{-1}$ will be the matrix representation
of $\pf{\log}{\mu}$, noting that the non-singularity of $\pf{\exp}{\log\mu}$
everywhere implies the invertibility of $Q$. If $\zeta_{j}\in\real^{d}$
is the coefficient vector (viewed as a column vector) of $\pf{\exp}{\log\mu}\phi_{j}$ with respect to a chosen basis, then it
is seen that the $Q$ is given by the matrix $Q=[\zeta_{1}\:\ldots\;\zeta_{d}]$ concatenated by column vectors $\zeta_j$. Therefore, the problem
boils down to the computation of 
\[
\pf{\exp}{\log\mu}\phi_{j}=\sum_{k=1}^{\infty}\frac{1}{k!}\sum_{\ell=0}^{k-1}(\log\mu)^{k-\ell-1}\cdot\phi_{j}\cdot(\log\mu)^{\ell}. 
\]
 Numerically, the above series is truncated
at a sufficiently large $K$. Note that when $\mu=I_{m}$, we have
specially $\pf{\exp}{\log\mu}=I_{m}$.


\subsection{A  Local Scan Procedure}

Roughly speaking, an ideal  statistic for detecting change patterns, or change points,   at a position $x$ should directly relate to the possibility that $x$ is a change point. 
The statistic at the position $x$  we proposed is a locally weighted average of the transformed $Y_i$'s near  $x$:
\[
G(x,h)=\sum_{i=1}^{n}w_{i}(h)\vec{\log Y_{i}}, 
\]
where {$w_{i}(h)=1/h$ if $1-h\leq i-x\leq0$, $w_{i}(h)=-1/h$ if $1\leq i-x\leq h$, and $w_i(h)=0$ otherwise.}
We remind the readers that $\vec{A}$ denotes the coefficient vector of the matrix $A$ with respect to a basis $\phi$.
The $G(x,h)$ defined above is constructed based on data points within a local window of size $2h$ around the point $ x $. The intuition is that, {if there is no change point within the window $(x-h,x+h)$, then $G(x,h)$ has mean zero and $\|G(x,h)\|_2$ is close to zero. Otherwise, if $\|G(x,h)\|_2$ is large, then $x$ is likely to be a change point. In particular, points that locally maximize $\|G(\cdot,h)\|_2$ have a high chance of being a change point.}   We say that $x$ is a $h$ local maximizer if $\lnorm{G(x,h)}\geq\lnorm{G(j,h)}$
for all $j\in\{x-h,\ldots,x+h\}$. The set of $h$ local maximizers
is denoted by $\mathcal{L}(h)$. Suppose $\lnorm{G(j_{1},h)}\geq\lnorm{G(j_{2},h)}\geq\cdots\geq\lnorm{G(j_{\ell},h)}$,
where $\ell$ is the number of elements in $\mathcal{L}(h)$. For
a given threshold $\rho>0$, we then estimate $\mathcal{J}$ by $\widehat{\mathcal{J}}=\{\tau\in\mathcal{L}(h):\lnorm[2]{G(\tau,h)}\geq\rho\}$,
and $J$ is estimated by the cardinality of $\widehat{\mathcal{J}}$. 


However, the above procedure for estimating $J$ depends on the unknown parameter $\rho$. 
In practice, we propose a data-driven alternative based on the 
$K$-fold cross validation to select  the number of change points. 
Suppose that  $j_{(1)},\ldots,j_{(k)}$ are change
points, which  divided all time points into $k+1$ segments. Within each segment, time points are randomly split  into $K$ partitions. The sample mean of a segment
is estimated by using data from any $K-1$ partitions within that segment, and the  validation
error is  evaluated on the rest one partition. 
The cross validation error of the segment is defined to be the sum of validation errors from the $K$ partitions, while
the total cross-validation error  is the sum of cross-validation errors across all $k+1$ segments. 
Formally, the total cross-validation error  is defined as 
\[
CV(k)=\sum_{q=1}^{k+1}\sum_{p=1}^{K}\sum_{i\in\mathscr{P}_{q,p}}(\widehat{\log\mu}_{q,p}-\vec{\log Y_{i}})^{\T}(\widehat{\log\mu}_{q,p}-\vec{\log Y_{i}}),
\]
where $\mathscr{P}_{q,p}$ is the $p$th partition of the $q$th segment, and $\widehat{\log\mu}_{q,p}=|\mathscr{P}_{q,-p}|^{-1}\sum_{i\in\mathscr{P}_{q,-p}}\vec{\log Y_{i}}$. Here $\mathscr{P}_{q,-p}$ denotes the time points in the $q$th
segment but not in the partition $\mathscr{P}_{q,p}$ and $|\mathscr{P}_{q,-p}|$
the cardinality of the set $\mathscr{P}_{q,-p}$. {The integer that minimizes $CV(\cdot)$ is chosen as an estimate of $J$.} We then estimate the locations of change points using the proposed  scan statistics. 

\section{Asymptotic Theory}\label{sec:3}

A random vector $\xi\in\real^{d}$ is called a subgaussian vector
with parameter $(\nu,\sigma)$ if $\eta\geq0$, $\nu\in\real^{d}$ and
for all $a\in\real^{d}$,
\[
\expect\left[\exp\left(a^{\T}(\xi-\nu)\right)\right]\leq\exp(\|a\|^{2}\sigma^{2}/2).
\]
We say that a random element in $\sym$ is subgaussian if its coefficient vector with respect to the orthonormal basis $\phi$ is subgaussian. One can easily check that this definition is independent of the choice of the  orthonormal basis of $\sym$. Below we shall assume $\log^\prime_{\mu_i}\varepsilon_i$ is subgaussian with parameter $(0,\sigma_i)$. This $\sigma_i$ might depend on $\mu_i$ and thus the linear operator $\log^\prime_{\mu_i}$ and its matrix representation $\Sigma_i$. For example, one might conceive of i.i.d. subgaussian random elements $\varepsilon_1,\ldots,\varepsilon_n$ and applying $\Sigma_i$ to $\varepsilon_i$, where different transformations result in different distributions of $\log^\prime_{\mu_i}\varepsilon_i$. Although subgaussianity is well preserved by linear transformations, the subgaussian parameter might differ after transformation. For example, one can show that, if $\varepsilon_i$ is subgaussian with a parameter $\theta$, then $\log^\prime_{\mu_i}\varepsilon_i$ is subgaussian, but with a parameter $\theta\sqrt{\|\Sigma_i\Sigma_i^{\T}\|}$. In this case, $\sigma_i$ might quantify the magnitude, measured by $\sqrt{\|\Sigma_i\Sigma_i^{\T}\|}$, of the transformation $\log^\prime_{\mu_i}$.

To derive the sure coverage property of the proposed procedure, we define $\delta=\inf\{\|\delta_{\tau}\|_{2}:\tau\in\mathcal{J}\}$, $\sigma=\max\{\sigma_1,\ldots,\sigma_n\}$, and $L=\inf_{1\leq j\leq J}(\tau_{j}-\tau_{j-1})$, where we conventionally
denote $\tau_{0}=0$ and $\tau_{J+1}=n$. 
We need the following assumption. 
{
\begin{assumption}\label{ass:1}
The quantities $\delta$, $L$, and  $\sigma$ satisfy that  $\delta^{2}L\geq 16\sigma^2(d+2\sqrt d+2\log n+2\log\log n)$.
\end{assumption}}
Here, $\sigma$  characterizes the variability of $\pf{\log}{\mu_i}\varepsilon_{i}$'s over all time
points. The quantity $\delta$
characterizes  the strength of the weakest signal of change points, while $L$ indicates the separability of change points. Intuitively, when $\delta$ and $L$ are small, no method would succeed in  recovering all change points. Recall that $d=m(m+1)/2$ denotes the dimension of the space $\spd$. It is seen that detection of change points becomes harder for higher dimensional matrices, i.e., a larger $m$, since  stronger signal (larger $\delta$) or  better separation of change points (larger $L$) is required to make the inequality in the above assumption hold.

Now we establish  the sure coverage property of the proposed procedure, that is, the union of the intervals selected by our procedure recovers all change points with probability going to $1$.   An nonasymptotic probability bound is also derived, with explicit dependence on the sample size $n$. 
We use  $\mathcal{J}\subset\widehat{\mathcal{J}}\pm h$
to denote that  $\tau_{j}\in(\hat{\tau}_{j}-h,\hat{\tau}_{j}+h)$ for all
$j=1,\ldots,J$. 
We are ready to state the main theorem of this paper, whose  proof
is deferred to the appendix. 
\begin{theorem}\label{thm:1}\label{thm:sure-coverage}
Suppose that Assumption \ref{ass:1} holds. If $\rho=\delta^{2}/4$ and $h=L/2$, then 
\$
\prob\big(\hat{J}&=J,\mathcal{J}\subset \widehat{\mathcal{J}}\pm h\big)\rightarrow 1, ~\textnormal{as}~  n\rightarrow \infty. 
\$
\end{theorem}

\begin{remark}
We emphasize here that the dimension $d$ does not need to be  assumed to be fixed
and could potentially diverge to infinity as long as Assumption \ref{ass:1} holds.
\end{remark}

\section{Simulation Studies}\label{sec:4}

In this section, we examine the empirical performance of our method. We generate data according to model \eqref{eq:model-sym}. In the first example, we consider different combinations of $n, m, J$ such that $ (n, m, J)=(100, 6, 2)$, $(200, 6, 2)$, $(200, 6, 4)$, $(400, 6, 4)$ respectively. When $ J=2 $, we set $ \mu_{1}=\ldots=\mu_{n/4}=I_m $, $ \mu_{n/4+1}=\ldots=\mu_{3n/4}=2I_m $ and $\mu_{3n/4+1}=\ldots=\mu_{n}=5I_m $. When $J=4$, we set $ \mu_{kn/5+1}=\ldots=\mu_{(k+1)n/5}=A_k $ for $ k=0, \ldots, 4$, where $ A_1=I_m $, $A_2={\rm diag}(I_{m/2}$, $3I_{m/2}) $, $A_3=3I_m $, $A_4={\rm diag}(3I_{m/2}, 10I_{m/2}) $, and $A_5=10I_m $. For the symmetric random noise, we first sample the coefficient vector from distribution $N(0,I_d)$, then combine it with the basis $\phi$ to generate the noise $\varepsilon_i$. Our second example is concerned with  a larger $m$ by setting that $ (n, m, J)=(100, 10, 2)$, $(200, 10, 2)$, $(200, 10, 4)$, $(400, 10, 4)$. 

Choosing the optimal bandwidth is  usually a difficult task for change point problems, see, for example, \citet{Niu2012} for a detailed discussion.  Intuitively, when there is only one change point in the interval $(x-h, x+h)$, the larger $h$ is, the more powerful the scan statistic is. But when the bandwidth gets too large, the interval might contain multiple change points. Therefore we need to choose bandwidth carefully. In our simulations, we found that
 the performance of the procedure is relatively robust to the choice of the bandwidth as long as the bandwidth is not too large, and $h=20$ works relatively well in our case.  We  use the proposed cross validation technique to select the number of change points. We run 100 repetitions of Monte Carlo studies. For each run, we  calculate the estimated number of change points and the locations of the change points. We report the frequencies of the three cases: $\widehat J\!<\!J $, $\widehat J\!=\!J$ and $\widehat J\!>\!J$, the mean of the number of change points detected, and the sure coverage probability for each of the change point. We also compare two methods that are frequently used in practice. {The first one vectorizes the response $Y_i$ without considering any manifold structure, which results in a $ m^2 $-dimensional vector. We denote this method as ``Vector''. The other one also adopts the vectorization idea, but additionally takes the symmetric information into account, yielding a $m(m+1)/2$-dimensional vector. We use ``Symmetric"  to denote this method.} The results are summarized in Tables \ref{table:estimation} and \ref{table:estimation2}. 
 

\begin{table}
{
\footnotesize
\begin{threeparttable}
\caption{The frequency of the number of change points when $\widehat J\!<\!J $, $\widehat J\!=\!J$ and $\widehat J\!>\!J$, the mean (s.e.) of the number of change points, and SCP of each change-point are reported. 
The results are based on 100 replications.}{%
\begin{tabular}{llccccccccc}
$ (n,m,J) $ & Method & $\widehat J\!<\!J $ & $\widehat J\!=\!J$ & $\widehat J\!>\!J$ & Mean & SCP 1 & SCP 2  & SCP 3 & SCP 4 \\ 
$ (100, 6, 2) $ & Proposed & 0.01 & 0.99 & 0 & 1.99(0.01) & 0.99 & 1 & NA & NA \\
                          & Vector  & 0.32 & 0.62 & 0.06 & 1.74(0.06) & 0.67 & 0.99 & NA & NA \\
                          & Symmetric  & 0.13 & 0.82 & 0.05 & 1.92(0.04) & 0.83 & 1 & NA & NA \\
$ (200, 6, 2) $ & Proposed & 0 & 0.95 & 0.05 & 2.05(0.02) & 1 & 1 & NA & NA \\
                          & Vector  & 0.64 & 0.22 & 0.14 & 1.56(0.09) & 0.17 & 0.99 & NA & NA \\
                          & Symmetric  & 0.65 & 0.17 & 0.18 & 1.62(0.1) & 0.22 & 1 & NA & NA \\
$ (200, 6, 4) $ & Proposed & 0.03 & 0.94 & 0.03 & 4(0.02) & 1 & 0.99 & 0.99 & 0.99\\
                          & Vector  & 0.56 & 0.27 & 0.17 & 3.08(0.14) & 0.47 & 0.48 & 0.93 & 0.89\\
                          & Symmetric  & 0.45 & 0.3 & 0.25 & 3.39(0.15) & 0.58 & 0.62 & 0.95 & 0.89\\
$ (400, 6, 4) $ & Proposed & 0 & 0.92 & 0.08 & 4.09(0.03) & 1 & 1 & 1 & 1\\
                          & Vector  & 0.6 & 0.23 & 0.17 & 3.28(0.18) & 0.05 & 0.1 & 0.73 & 0.89\\
                          & Symmetric  & 0.55 & 0.25 & 0.2 & 3.59(0.2) & 0.1 & 0.15 & 0.85 & 0.88\\
\end{tabular}}
\label{table:estimation}
\begin{tablenotes}
\footnotesize
\item  s.e., standard error; SCP, sure coverage probability; NA, not available.
\end{tablenotes}
\end{threeparttable}}
\end{table}

As indicated by the results, our proposed method performs better than the comparison methods in terms of the percentage of correctively recovering the number of change points, in all cases. Additionally, when $J=2$, all the methods presented here can detect the second change point very well. But for the first one, our method achieves a higher sure coverage probability. When $J=4$, our method achieves a higher sure coverage probability for all the change points. These results suggest the importance of considering the geometric structure of the Riemannian data, at least,  in change point detection problems.

\begin{table}
\footnotesize
\begin{threeparttable}
\caption{The frequency of the number of change points when $\widehat J\!<\!J $, $\widehat J\!=\!J$ and $\widehat J\!>\!J$, the mean (s.e.) of the number of change points, and the SCP of each change-point are reported. 
The results are based on 100 replications.}{%
\begin{tabular}{llccccccccc}
$ (n,m,J) $ & Method & $\widehat J\!<\!J $ & $\widehat J\!=\!J$ & $\widehat J\!>\!J$ & Mean & SCP 1 & SCP 2  & SCP 3 & SCP 4 \\ 
$ (100, 10, 2) $ & Proposed  & 0.01 & 0.99 & 0 & 1.99(0.01) & 0.99 & 1 & NA & NA \\
                          & Vector  & 0.61 & 0.38 & 0.01 & 1.4(0.05) & 0.37 & 1 & NA & NA \\
                          & Symmetric  & 0.27 & 0.69 & 0.04 & 1.77(0.05) & 0.72 & 1 & NA & NA \\
$ (200, 10, 2) $ & Proposed & 0 & 0.98 & 0.02 & 2.02(0.01) & 1 & 1 & NA & NA \\
                          & Vector  & 0.67 & 0.17 & 0.16 & 1.51(0.08) & 0.17 & 0.97 & NA & NA \\
                          & Symmetric  & 0.74 & 0.11 & 0.15 & 1.45(0.09) & 0.21 & 1 & NA & NA \\
$ (200, 10, 4) $ & Proposed  & 0.04 & 0.96 & 0 & 3.96(0.02) & 1 & 0.96 & 1 & 1\\
                          & Vector  & 0.51 & 0.31 & 0.18 & 3.28(0.12) & 0.33 & 0.57 & 0.94 & 0.96\\
                          & Symmetric  & 0.44 & 0.36 & 0.2 & 3.38(0.13) & 0.51 & 0.65 & 0.96 & 0.96\\
$ (400, 10, 4) $ & Proposed & 0 & 1 & 0 & 4(0) & 1 & 1 & 1 & 1\\
                          & Vector  & 0.64 & 0.21 & 0.15 & 2.95(0.19) & 0.03 & 0.08 & 0.55 & 0.89\\
                          & Symmetric  & 0.51 & 0.25 & 0.24 & 3.66(0.22) & 0.1 & 0.17 & 0.76 & 0.93\\
\end{tabular}}
\label{table:estimation2}
\begin{tablenotes}
\item s.e., standard error; SCP, sure coverage probability; NA, not available.
\end{tablenotes}
\end{threeparttable}
\end{table}

\section{An Application to the Human Connectome Project}\label{sec:5}
We apply the proposed methodology to the 
social cognition task related fMRI data from Human Connectome Project Dataset, which includes behavioral and 3T MR imaging data from 970 healthy adult participants collected from 2012 to spring 2015. We focus on the 850 subjects out of the 970 which have the 
social cognition task related fMRI data. 
Participants were presented with short video clips (20 seconds) of objects (squares, circles, triangles) that either interacted in some way, or moved randomly on the screen \citep{Castelli2000, Wheatley2007}. 
There were 5 video blocks (2 Mental and 3 Random in one run, 3 Mental and 2 Random in the other run) in the task run. 

We use the ``Desikan-Killiany'' atlas \citep{Desikan2006} to divide the brain into 68 regions of interest. Figure \ref{fig:ROIs}(a) shows the ``Desikan-Killiany'' parcellation of the cortical surface in the left and right hemisphere. We pick eight possible regions that are related to the social task, that are the left and right part of superior temporal, inferior parietal, temporal pole and precuneus \citep{Green2015}. 
These eight regions of interest are highlighted in yellow in Figure \ref{fig:ROIs}(b). For each subject, the fMRI data are recorded on 274 evenly spaced time points, one per 0.72 seconds. We use a moving local window of size 100 to calculate the cross covariance between these eight regions, which results in 175 cross covariance matrices with dimensions $ 8\times 8 $.  We then apply the proposed method to detect change points in this sequence of cross covariance matrices with bandwidth $h$ set to be $20$.

 We apply the method  to all the subjects, and report the locations of change points detected for each subject. In Table \ref{table:data}, we have summarized the count and percentage of the number of change points detected for all the subjects. The mean number of change points detected among all the subjects is $3.66(0.03)$. This result matches the physiology well since there are 5 video blocks in the task design, with changes  at the time points $ 35 $, $70$, $105 $ and $ 140 $ respectively.  
To further validate the proposed methodology,  we pick up all those subjects with four change points, and calculate the mean locations respectively. The means are $39.6$, $72.0$, $106.7$, $138.5$, which are fairly close to  the task block changes.  A more interesting observation is that the lags (4.6s, 2.0s, 1.7s and -1.5s) are becoming shorter and shorter: the last time point even precedes the designed time point. This, to our understanding, witnesses the powerfulness of human brains for learning the change patterns.



\begin{table}[th]
  \centering
    \caption{The count and percentage of the number of change points detected for social task related fMRI data.}
\vspace{3pt}
\begin{tabular}{lccccc}
 Number of Change Points & 1 & 2 & 3 & 4 & 5 \\ 
Count & 1 & 55 & 291 & 387 & 116 \\
Percentage & 0.12\% & 6.47\% & 34.24\% & 45.53\% & 13.65\% \\                                                                                                    
\end{tabular}
\label{table:data}
\end{table}

 \begin{figure}
  \centering
    \subfigure[]{\includegraphics[width=0.4\textwidth]{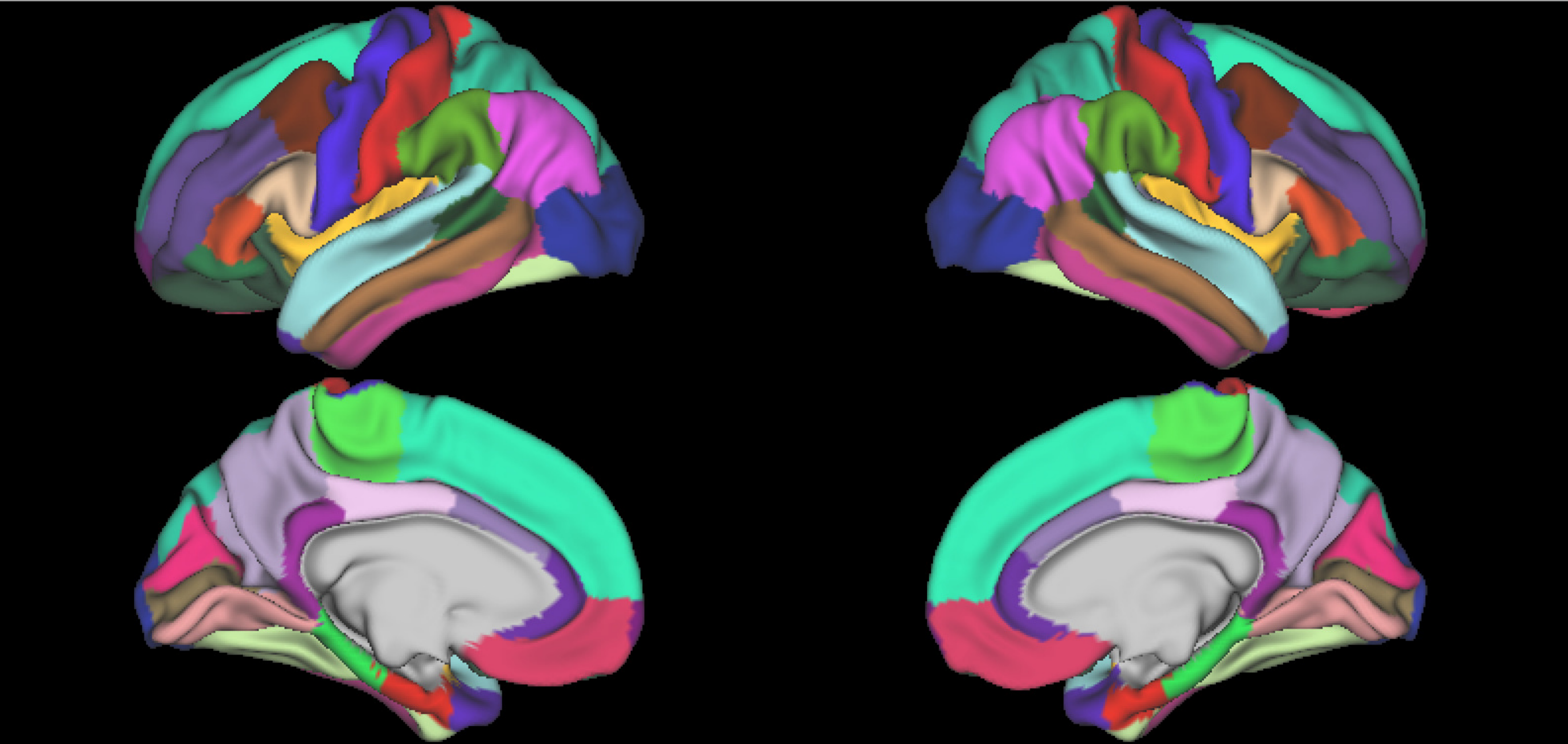}}
       \subfigure[]{\includegraphics[width=0.4\textwidth]{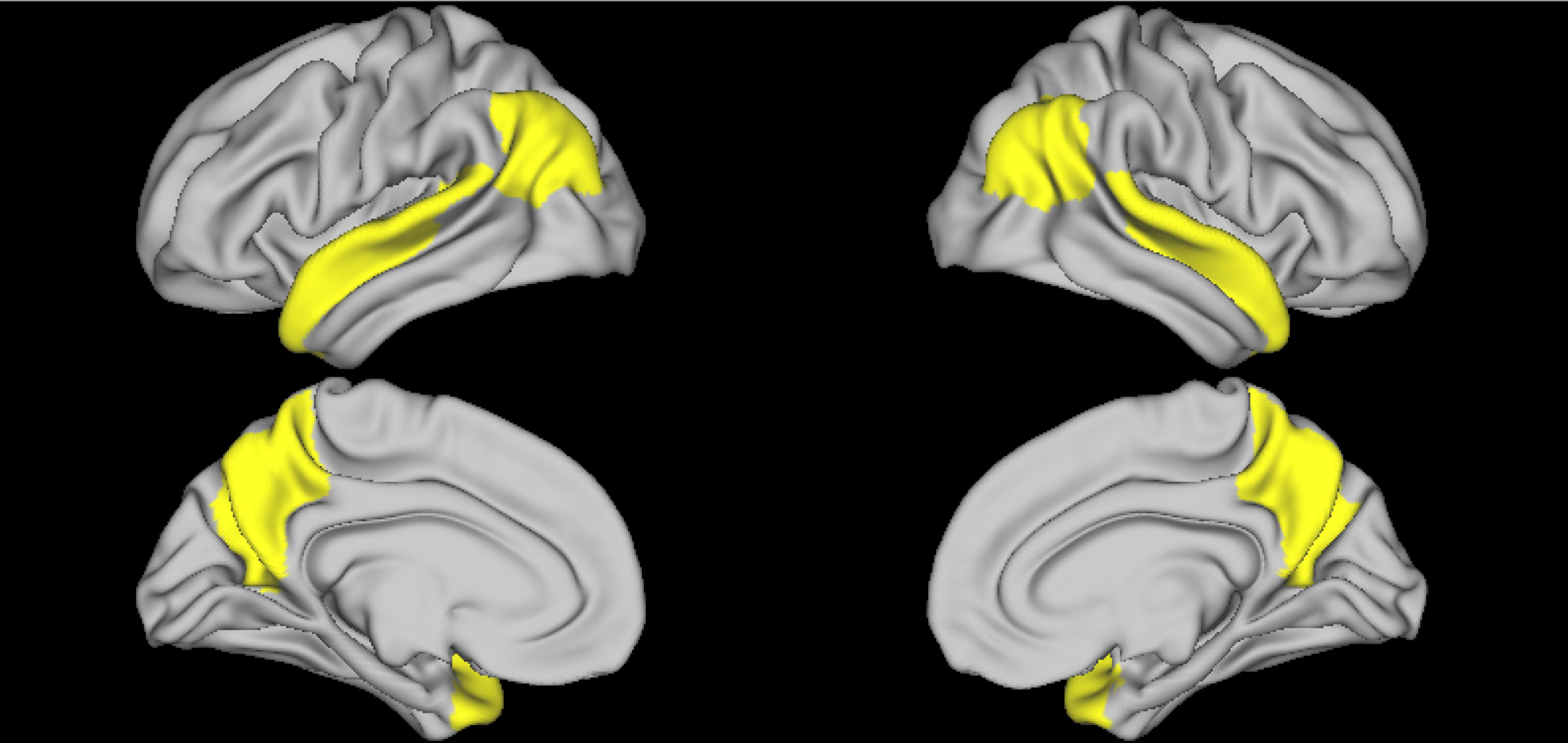}}                   
\caption{Panel (a): the ``Desikan-Killiany'' parcellation of the cortical surface in left and right hemisphere. Panel (b): the eight regions of interest.}
\label{fig:ROIs}
\end{figure}

\section{Discussion}

In this paper we propose an additive matrix-log mean model with a heterogeneous noise for modeling random symmetric positive definite matrices that lie in a Riemannian manifold.  The heterogeneous noise part takes account the manifold structure of the original symmetric positive definite matrices. 
Built upon this model, 
 we then propose a scan statistic to perform multiple change point detection.  Theoretical studies and numerical examples lend further support to our proposed methodology.
 
Our proposed methodology replies on the assumption that the collected samples $Y_i$'s are independent.  Independence is an  ideal assumption that may be violated in some settings. 
However, this assumption allows us to conduct theoretical analysis, which also produce results that could be useful when the assumption is  violated. We will pursue the change-point detection problems  under dependence for  Riemannian data in future work. 



\bibliographystyle{ims}
\bibliography{misc,manifold}

\newpage
\appendix 
\renewcommand{\theequation}{S.\arabic{equation}}
\renewcommand{\thetable}{S.\arabic{table}}
\renewcommand{\thefigure}{S.\arabic{figure}}
\renewcommand{\thesection}{S.\arabic{section}}
\renewcommand{\thelemma}{S.\arabic{lemma}}

\vspace{30pt}
\noindent{\bf \LARGE Appendix}
\vspace{-10pt}
\section{Preliminary}
In this section, we further discuss the smooth and Riemannian
manifold. For a comprehensive treatment on these subjects,
readers are referred to the introductory book by 
\citet{Lee1997}. 

A smooth manifold $\manifold$ is a differentiable manifold with all
transition maps being $C^{\infty}$-differentiable. Associated with
each point $x$ on the manifold $\manifold$, there exists a linear space $\tangentspace x$
called the tangent space at the base point $x$. Each element in the tangent space is
called a tangent vector. For a manifold that is a  submanifold of a Euclidean
space, the tangent space at a point can be geometrically visualized
as the hyperplane tangent to that point, while tangent vectors are visualized as Euclidean vectors tangent to the manifold at that point; see Figure \ref{fig:tangent} for an illustration. It is emphasized that {tangent vectors at different base points are different, despite that the vectors might point to the same direction}. Thus, a tangent vector always implicitly comes with a base point. 
For Euclidean submanifolds, a tangent vector $v$ at a point $x$ can also be algebraically interpreted as a  directional derivative $D_{v}$ at $x$, such that $D_{v}f=v^\T\nabla f(x)$ for all $f\in C^\infty(\manifold)$ with $C^\infty(\manifold)$ denoting the collection of real-valued smooth functions  defined on the manifold $\manifold$. Observe that $D_v$ is a {derivation} at $x$, which satisfies the Leibniz rule, 
$$D_v(fg)=g(x)(D_vf)+f(x)(D_vg),$$ for any  $f,g\in C^\infty(\manifold)$ and $v\in \tangentspace x$, where $fg$ denotes the pointwise product of functions. This allows one to generalize the concept of tangent vector as directional derivative to non-Euclidean manifolds, by defining tangent vectors at $x$ as derivations at $x$, and tangent space at $x$ as the space of derivations at $x$. A convenient way to perceive the derivation represented by a tangent vector $v$ is to treat $D_v:\,C^\infty(\manifold)\rightarrow \real$ as a linear functional that maps $C^\infty(\manifold)$ into $\real$. 

For a smooth transformation $\varphi:\mathcal{N}\rightarrow\manifold$ that maps points on a manifold $\mathcal{N}$ to points on the manifold $\manifold$, its \emph{differential}
at $x$, denoted by $\pf{\varphi}{x}$, is a linear map sending a tangent vector  $v\in\tangentspace[\mathcal{N}]x$
to a tangent vector $\pf{\varphi}{x}(v)\in\tangentspace[\manifold]{\varphi(x)}$, such that the derivation $D_{\pf{\varphi}{x}(v)}$ corresponding to the tangent vector $\pf{\varphi}{x}(v)$  at the point $\varphi(x)\in\manifold$  is depicted by $$D_{\pf{\varphi}{x}(v)}: C^\infty(\manifold)\rightarrow\real,\,\,\,\text{s.t.}\,\,\,\forall f\in C^\infty(\manifold):\,\,D_{\pf{\varphi}{x}(v)}f=D_{v}(f\circ\varphi),$$ where $f\circ g$ denotes the composition of functions. 
 When both $\mathcal{N}$ and $\manifold$ are Euclidean submanifolds, the differential $\pf{\varphi}{x}$ is the usual notion of differential of the function $\varphi$ at $x$, given by a Jacobian matrix.  Specially, when $\mathcal{N}$ is an interval $(a,b)$ of the real line, the tangent space at each $t\in(a,b)$ is the whole real line $\real$. In this case, $\varphi$ is often called a (parameterized) smooth curve on $\manifold$, and $\pf{\varphi}{t}(1)$  is denoted by $\varphi^\prime(t)$ which is the derivative of the curve at time $t$. Here, to properly decode the notation $\pf{\varphi}{t}(1)$,  recall that $\pf{\varphi}{t}$ is a map sending a tangent vector of a manifold to a tangent vector of another manifold, and for the special manifold $(a,b)$, the real number $1$ can be viewed as a tangent vector at $t\in(a,b)$. Geometrically and intuitively, $\varphi^\prime(t)$ is a vector tangent to the curve $\varphi$ at time $t$, as illustrated in Figure \ref{fig:tangent}.


A Riemannian manifold is a smooth manifold endowed with an inner product
$\innerprod{\cdot}{\cdot}_{x}$ on the
tangent space at each point $x$, such that $\innerprod{\cdot}{\cdot}_{x}$
varies with $x$ smoothly. The collection of such inner products is
often called the Riemannian metric tensor, or simply Riemannian metric.
One can show that, the metric tensor induces a distance function that
turns the manifold into a metric space. A geodesic is a smooth curve on the manifold such that for
any sufficiently small segment, the segment is the unique smooth
curve with the minimal length among all smooth curves connecting the two
endpoints of the segment. Every smooth curve on the manifold can be
parameterized by a smooth map $\gamma$ from an interval in $\real$
to the manifold. For any $u\in\tangentspace x$, there exists a unique
geodesic $\gamma$ such that $\gamma^\prime(0)=u$ and $\gamma(0)=x$.
Then the exponential map at $x$, denoted by $\Exp_{x}$, is defined
by $\Exp_{x}(u)=\gamma(1)$. For example, one can verify that for
the unit circle $\mathbb{S}^{1}\equiv\{(z_1,z_2)\in\real^{2}:z_1^{2}+z_2^{2}=1\}$,
for $x\in\mathbb{S}^{1}$ and $v\in\tangentspace[\mathbb{S}^{1}]x$,
the defining geodesic for the Riemannian exponential map $\Exp_{x}$
is $\gamma(t)=\cos(t\|u\|_{2})x+\sin(t\|u\|_{2})u/\|u\|_{2}$, as
$\gamma(0)=x$ and $\gamma^\prime(0)=u$. Thus, $\Exp_{x}(u)=\gamma(1)=\cos(\|u\|_{2})x+\sin(\|u\|_{2})u/\|u\|_{2}$. A graphical illustration of  the Exp map is given in Figure \ref{fig:tangent}.

\begin{figure}
\center{\includegraphics[scale=1]{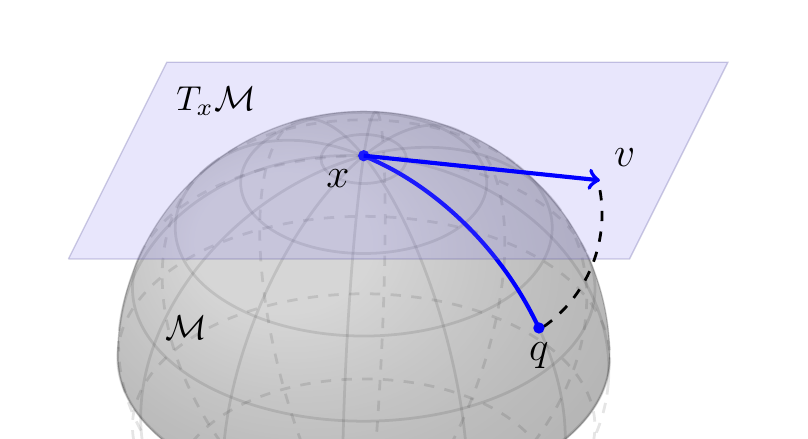}}
\caption{Geometric illustration of tangent vector, tangent space, curve and exponential map. $\gamma(t)$ with $\gamma(0)=x$ and $\gamma(1)=q$ is a smooth curve on $\manifold$. $v$ is a tangent vector at $x$ and also tangent to the curve $\gamma$ at $t=0$, i.e. $v=\gamma^\prime(0)$. If in addition $\gamma(t)$ is a geodesic, then $q=\gamma(1)=\Exp_x(v)$. }\label{fig:tangent}
\end{figure}

For the Log-Euclidean metric,
at the identity matrix $I$, it is defined as $\innerprod UV_{I}=\tr(UV)$ 
for $U,V\in\sym=\tangentspace[\spd]I$,   the Frobenius inner
product on $\sym$. In order to define metric at other points, the
following group structure is considered. Define $S_{1}\odot S_{2}=\exp(\log S_{1}+\log S_{2})$,
where $\exp$ and $\log$ are the matrix exponential and  logarithm  respectively. 
 The operation $\odot$ turns $\spd$ into a group.
Now we define the left-translation operator $\mathscr{L}_{S}Q=S\odot Q$
for $S,Q\in\spd$. As shown in \citet{Arsigny2007}, $\mathscr{L}_{S}$
is a smooth map from $\manifold$ to itself. Thus, its differential
$\pf{(\mathscr{L}_{S})}{Q}$ at $Q$ is a linear map that sends tangent
vectors at $S$ to tangent vectors at $S\odot Q$. For instance, the
linear operator $\mathscr{I}_{S}=\pf{(\mathscr{L}_{S^{-1}})}{S}$ maps tangent vectors at $S$ to tangent vectors at the identity. Given
this property, we can ``translate'' the metric at the identity matrix
to all points by the left-translation operator $\mathscr{L}$. More
specifically, the Log-Euclidean metric at any $S\in\spd$ is defined 
by $\innerprod UV_{S}=\innerprod{\mathscr{I}_{S}U}{\mathscr{I}_{S}V}_I=\innerprod{\pf{\log}{S}U}{\pf{\log}{S}V}_I$
for all $U,V\in\sym$, where the last identity is due to $\pf{(\mathscr{L}_{S^{-1}})}{S}=\pf{\log}S$ \citep{Arsigny2007}. The Riemannian exponential map under this metric
is given by $\Exp_{S}U=\exp(\log S+\pf{\log}{S}U)$.


\section{ Proofs}
\begin{lemma}\label{lem:2}
Let $\mathcal{A}_{n}(h,\rho)=\bigcap_{x\in\mathcal{F}}\{\lnorm[2]{G(x,h)}\leq\rho\}$,
$\mathcal{B}_{n}(h,\rho)=\bigcap_{\tau\in\mathcal{J}}\{\lnorm[2]{G(\tau,h)}\geq\rho\}$
and $\mathcal{E}_{n}(h,\rho)=\mathcal{A}_{n}(h,\rho)\cap\mathcal{B}_{n}(h,\rho)$, where $\mathcal{F}$ denotes the collection of flat points, i.e., $x\in\mathcal{F}$ if and only if $\mu_{j}=\mu_{x}$ for all $j\in\{x-h+1,\ldots,x+h\}$. Then $\mathcal{J}\subset \widehat{\mathcal{J}}\pm h$ holds under the event $\mathcal{E}_n(h,\rho)$.
\end{lemma}
\begin{proof}[Proof of Lemma \ref{lem:2}]
The proof can be found in Lemma 3 of \cite{Niu2012}.
\end{proof}


\begin{proof}[Proof of Theorem \ref{thm:sure-coverage}]
We first note the following facts about subgaussian random vectors
that will be used in the sequel.
\begin{itemize}
\item If $\xi_{1},\ldots,\xi_{n}$ are independent and  subgaussian 
with parameters $(\mu_{1},\eta_{1}),\ldots,(\mu_{n},\eta_{n})$, respectively,
then $\sum_{i}\xi_{i}$ is a subgaussian random vector with parameters
$\sum_{i}\mu_{i}$ and $\{\sum_{i}\eta_{i}^{2}\}^{1/2}$.
\item If $\xi\in\real^{d}$ is a subgaussian random vector with a parameter
$(\mu,\eta)$ and $A$ is a $d\times d$ matrix, then $A\xi$ is a
subgaussian random vector with parameters $A\mu$ and $\eta\sqrt{\|AA^{\T}\|}$.
\end{itemize}

A point $x$ is called a $h$-flat point (or simply flat point if
$h$ is clear from the context) if $\mu_{j}=\mu_{x}$ for all $j\in\{x-h,\ldots,x+h\}$.
For a flat point $x$, $G(x,h)=\sum_{i=x-h}^{x}h^{-1}\Sigma_{i}\vec{\varepsilon}_{i}-\sum_{i=x+1}^{x+h}h^{-1}\Sigma_{i}\vec{\varepsilon}_{i}$
has mean zero and also is a subgaussian random vector with parameters
$0$ and $\sqrt{2h^{-2}\sum_{i}\sigma_{i}^{2}}\leq\sigma\sqrt{2/h}$.
Here, we recall that $\Sigma_{i}\vec{\varepsilon}_{i}$ is subgaussian
with the parameter $(0,\sigma_{i})$. For a change-point $\tau$,
similarly, $G(\tau,h)$ is subgaussian with parameter $(\delta_{\tau},\sigma\sqrt{2/h})$.

Let $t_{n}=\log n+\log\log n$ and $a_{n}=2\sigma^{2}h^{-1}(d+2\sqrt{d}+2t_n)$.
For a flat point $x$, we first observe that
\[
\prob\{\|G(x,h)\|_{2}^{2}>a_{n}\}\leq e^{-t_{n}}=\frac{1}{n\log n}
\]
according to Theorem 1 of  \cite{Hsu2012}. According to Assumption
\ref{sec:1} and the choice of $\rho$ and $L$, we have $\rho\geq a_{n}$ and
thus 
\[
\prob\{\|G(x,h)\|_{2}^{2}>\rho\}\leq\prob\{\|G(x,h)\|_{2}^{2}>a_{n}\}\leq\frac{1}{n\log n}.
\]
Similarly, from Assumption \ref{sec:1} we deduce that $\|\delta_{\tau}\|_{2}\geq\delta\geq2\sqrt{a_{n}}$
and $\sqrt\rho\leq\|\delta_{\tau}\|_{2}-\sqrt{a_{n}}$. Thus, for
a change point $\tau$,
\begin{align*}
\prob\{\|G(\tau,h)\|_{2}<\sqrt{\rho}\} & \leq\prob\{\|G(\tau,h)\|_{2}<\|\delta_{\tau}\|_{2}-\sqrt{a_{n}}\}\\
 & \leq\prob\{\left|\|G(\tau,h)\|_{2}-\|\delta_{\tau}\|_{2}\right|>\sqrt{a_{n}}\}\\
 & \leq\prob\{\|G(\tau,h)-\delta_{\tau}\|_{2}^{2}>a_{n}\}\\
 & \leq\frac{1}{n\log n},
\end{align*}
or equivalently, 
\[
\prob\{\|G(\tau,h)\|_{2}^{2}<\rho\}\leq\frac{1}{n\log n}.
\]

Next, we bound the probabilities of events defined in Lemma \ref{lem:2}.
\begin{align}
\prob\{\mathcal{E}_{n}(h,\rho)\} & =1-\prob\left[\{\mathcal{A}_{n}(h,\rho)\}^{c}\cup\{\mathcal{B}_{n}(h,\rho)\}^{c}\right] \nonumber \\
& \geq1-\prob\left[\{\mathcal{A}_{n}(h,\rho)\}^{c}\right]-\prob\left[\{\mathcal{B}_{n}(h,\rho)\}^{c}\right].\label{eq:EF-1}
\end{align}
Now, note that 
\begin{align}
\prob\left[\{\mathcal{A}_{n}(h,\rho)\}^{c}\right] & =\prob\left\{ \exists x\in\mathcal{F}:\lnorm[2]{G(x,h)}>\rho\right\} \leq\sum_{x\in\mathcal{F}}\prob\left\{ \lnorm[2]{G(x,h)}>\rho\right\} \nonumber \\
 & \leq\sum_{x\in\mathcal{F}}\frac{1}{n\log n}\leq n\left(\frac{1}{n\log n}\right)=\frac{1}{\log n}.\label{eq:EF-2}
\end{align}
Similarly, 
\begin{align}
\prob\left[\{\mathcal{B}_{n}(h,\rho)\}^{c}\right] & =\prob\left\{ \exists\tau\in\mathcal{\mathcal{J}}:\lnorm[2]{G(\tau,h)}<\rho\right\} \leq\sum_{\tau\in\mathcal{\mathcal{J}}}\prob\left\{ \lnorm[2]{G(\tau,h)}<\rho\right\} \nonumber \\
 & \leq\sum_{\tau\in\mathcal{J}}\frac{1}{n{\log n}}\leq\frac{n}{n{\log n}}=\frac{1}{{\log n}}.\label{eq:EF-3}
\end{align}
Combining (\ref{eq:EF-1}), (\ref{eq:EF-2}) and (\ref{eq:EF-3}),
we conclude that 
\begin{equation}
\prob\{\mathcal{E}_{n}(h,\rho)\}\geq1-\frac{2}{\log n}\rightarrow1.\label{eq:EF-4}
\end{equation}
Finally, the theorem follows from Lemma \ref{lem:2}. 
\end{proof}

\end{document}